\documentclass[11pt]{article}
\usepackage{fullpage}
\usepackage{citesort}
\usepackage{times}
\usepackage{amssymb, amsmath}
\usepackage{graphicx}
\allowdisplaybreaks

\newbox\ProofSym
\setbox\ProofSym=\hbox{%
\unitlength=0.18ex%
\begin{picture}(10,10)
\put(0,0){\framebox(9,9){}}
\put(0,3){\framebox(6,6){}}
\end{picture}}

\newtheorem{definition}{Definition}[section]
\newtheorem{theorem}{Theorem}[section]
\newtheorem{lemma}{Lemma}[section]

\newtheorem{claim}{Claim}[section]

\newcommand{\eps}{\varepsilon}
\newcommand{\polylog}{\mathrm{polylog}}
\newcommand{\E}{\mathsf{E}}

\newcommand{\A}{\mathcal{A}}
\newcommand{\C}{\mathcal{C}}
\newcommand{\var}{\mathsf{Var}}
\renewcommand{\Pr}{\mathsf{Pr}}

\newcommand{\abs}[1]{\left| #1 \right|}

\newcommand{\qinomit}[1]{}

\begin{document}
\title{Randomized Algorithms for Tracking Distributed Count, Frequencies, and
Ranks}

\author{Zengfeng Huang \hspace{1cm} Ke Yi \\
\\
Hong Kong University of Science and Technology \\
$\{$huangzf, yike$\}$@cse.ust.hk
\and Qin Zhang \\
\\
MADALGO, University of Aarhus \\
qinzhang@cs.au.dk
}

\maketitle

\begin{abstract}
  We show that randomization can lead to significant improvements for a few
  fundamental problems in distributed tracking.  Our basis is the {\em
    count-tracking} problem, where there are $k$ players, each holding a
  counter $n_i$ that gets incremented over time, and the goal is to track
  an $\eps$-approximation of their sum $n=\sum_i n_i$ continuously at all
  times, using minimum communication.  While the deterministic
  communication complexity of the problem is $\Theta(k/\eps \cdot \log N)$,
  where $N$ is the final value of $n$ when the tracking finishes, we show
  that with randomization, the communication cost can be reduced to
  $\Theta(\sqrt{k}/\eps \cdot \log N)$.  Our algorithm is simple and uses
  only $O(1)$ space at each player, while the lower bound holds even
  assuming each player has infinite computing power.  Then, we extend our
  techniques to two related distributed tracking problems: {\em
  frequency-tracking} and {\em rank-tracking}, and obtain similar
  improvements over previous deterministic algorithms.  Both problems are
  of central importance in large data monitoring and analysis, and have
  been extensively studied in the literature.
\end{abstract}


\section{Introduction}

We start with a very basic problem in distributed tracking, what we call
{\em count-tracking}.  There are $k$ players each holding a counter $n_i$
that is initially $0$.  Over time, the counters get incremented and we denote
by $n_i(t)$ the value of the counter $n_i$ at time $t$.  The goal is to
track an $\eps$-approximation of the total count $n(t)=\sum_i n_i(t)$,
i.e., an $\hat{n}(t)$ such that $(1-\eps)n(t) \le \hat{n}(t) \le
(1+\eps)n(t)$,\footnote{We sometimes omit ``$(t)$'' when the context is
  clear.} continuously at all times.  There is a coordinator whose job is
to maintain such an $\hat{n}(t)$, and will try to do so using minimum
communication with the $k$ players (the formal model of computation will be
defined shortly).

There is a trivial solution to the count-tracking problem: Every time a
counter $n_i$ has increased by a $1+\eps$ factor, the player informs the
coordinator of the change.  Thus, the coordinator always has an
$\eps$-approximation of every $n_i$, hence an $\eps$-approximation of their
sum $n$.  Letting $N$ denote the final value of $n$, simple analysis shows
that the communication cost of this algorithm is $O(k/\eps \cdot \log
N)$\footnote{A more careful analysis leads to a slightly better bound of
  $O(k/\eps \cdot \log(\eps N/k))$, but we will assume that $N$ is
  sufficiently large, compared to $k$ and $1/\eps$, to simplify the
  bounds.}.  This algorithm was actually used in \cite{keralapura06} for
solving essentially the same problem, which also provided many practical
motivations for studying this problem.  Note that this algorithm is
deterministic and only uses one-way communication (from the players to the
coordinator), and yet it turns out this simple algorithm is already optimal
for deterministic algorithms, even if two-way communication is allowed
\cite{yi09:_optim}.  Thus the immediate questions are: What about
randomized algorithms that are allowed to fail with a small probability?
Is two-way communication not useful at all?  In this paper, we set out to address these
questions, and then move on to consider other related distributed tracking
problems.

\subsection{The distributed tracking model}
We first give a more formal definition of the computation model that we
will work with, which is essentially the same as those used in prior work
on distributed tracking
\cite{Babcock:Olston:03,cormode08:algorithms,cormode10:_optim,Cormode:Garofalakis:Muthukrishnan:Rastogi:05,yi09:_optim,sharfman08:_shape,chan10:_contin,arackaparambil09,keralapura06}.
There are $k$ distributed {\em sites} $S_1,\dots,S_k$, each receiving a
stream of elements over time, possibly at varying rates.  Let $N$ be the
total number of elements in all $k$ streams.  We denote by $A_i(t)$ the
multiset (bag) of elements received by $S_i$ up until time $t$, and let
$A(t) = \biguplus_{i=1}^k A_i(t)$ be the combined data set, where $\uplus$
denotes multiset addition.  There is a coordinator whose job is to maintain
(an approximation of) $f(A(t))$ continuously at all times, for a given
function $f$ (e.g., $f(A(t)) = |A(t)|$ for the count-tracking problem
above).  The coordinator has a direct two-way communication channel with
each of the sites; note that broadcasting a message costs $k$ times the
communication for a single message.  The sites do not communicate with each
other directly, but this is not a limitation since they can always pass
messages via the coordinator.  We assume that communication is instant,
i.e., no element will arrive until all parties have decided not to send
more messages.  As in prior work, our measures of complexity will be the
communication cost and the space used to process each stream.  Unless
otherwise specified, the unit of both measures is a {\em word}, and we
assume that any integer less than $N$, as well as an element from the
stream, can fit in one word.

\begin{table*}[t]
  \centering
  \begin{tabular}{c|c|c|c}
\hline
& & space (per site) & communication \\
\hline
\hline
count-tracking & trivial & $O(1)$ &
$\Theta(k/\eps\cdot\log N)$ \\
& new & $O(1)$ & $O(\sqrt{k}/\eps\cdot \log N)$ \\
& & & $\Omega(\sqrt{k}/\eps\cdot \log N)$ messages \\
\hline
frequency-tracking & \cite{yi09:_optim} &  $O(1/\eps)$&
$\Theta(k/\eps \cdot \log N)$ \\
& new & $O(1/(\eps\sqrt{k}))$ & $O(\sqrt{k}/\eps \cdot \log N)$ \\
& & $\Omega(1/(\eps\sqrt{k}))$ bits$^\star$ & $\Omega(\sqrt{k}/\eps\cdot
\log N)$ messages \\ 
\hline
rank-tracking & \cite{yi09:_optim} &  $O(1/\eps \cdot \log n)$&
$O(k/\eps\cdot \log N\log^2(1/\eps))$ \\ 
& new &  $O\left(1/(\eps\sqrt{k})\cdot\log^{1.5}\frac{1}{\eps}\log^{0.5}\frac{1}{\eps
  \sqrt{k}}\right)$& $O\left(\sqrt{k}/\eps \cdot \log N
\log^{1.5}\frac{1}{\eps\sqrt{k}}\right)$ \\
&& $\Omega(1/(\eps\sqrt{k}))$ bits$^\star$ & $\Omega(\sqrt{k}/\eps\cdot \log N)$ messages 
\\
\hline
sampling & \cite{cormode10:_optim} & $O(1)$ & $O(1/\eps^2 \cdot \log N)$ \\
\hline
  \end{tabular}
  \caption{Space and communication costs of previous and new algorithms.
    We assume $k\le 1/\eps^2$.  All upper bounds are in terms of
    words. $^\star$This is conditioned upon the communication cost being
    $O(\sqrt{k}/\eps \cdot\log N)$ bits.}
  \label{tab:results}
\end{table*}

This model was initially abstracted from many applied settings, ranging
from distributed data monitoring, wireless sensor networks, to network
traffic analysis, and has been extensively studied in the database
community.  From 2008 \cite{cormode08:algorithms}, the model has started to
attract interests from the theory community as well, as it naturally
combines two well-studied models: the data stream model and multi-party
communication complexity.  When there is only $k=1$ site who also plays the
role of the coordinator, the model degenerates to the standard streaming
model; when $k\ge 2$ and our goal is to do a one-shot computation of
$f(A(\infty))$, then the model degenerates to the (number-in-hand)
$k$-party communication model.  Thus, distributed tracking is more general
than both models.  Meanwhile, it also appears to be significantly different
from either, with the above count-tracking problem being the best example.
This problem is trivial in both the streaming and the communication model
(even computing the exact count is trivial), whereas it becomes nontrivial in
the distributed tracking model and requires new techniques, especially when
randomization is allowed, as illustrated by our results in this paper.

Note that there is some work on {\em distributed streaming} (see e.g.\ 
\cite{lall10:_global,manjhi05:_findin,gibbons01,gibbons02:_distr}) that
adopts a model very similar to ours, but with a fundamental difference.  In
their model there are $k$ streams, each of which runs a streaming algorithm
on its local data.  But the function $f$ on the combined streams is
computed only at the end or upon requests by the user.  As one can see that
the count-tracking problem is also trivial in this model.  The crucial
difference is that, in this model, the sites wait passively to get polled.
If we want to track $f$ continuously, we have to poll the sites all the
time.  Whereas in our model, the sites actively participate in the tracking
protocol to make sure that $f$ is always up-to-date.

\subsection{Problem statements, previous and new results}
In this paper, we first study the count-tracking problem.  Then we extend
our approach to two related, more general problems: {\em
  frequency-tracking} and {\em rank-tracking}.  Both problems are of
central importance in large data monitoring and analysis, and have been
extensively studied in the literature.  In all the communication upper
bounds, we will assume $k\le 1/\eps^2$; otherwise all of them will carry an
extra additive $O(k \log N)$ term.  There are other good reasons to justify
this assumption, which we will explain later.  All our results are
summarized in Table~\ref{tab:results}; below we discuss each of them
respectively.

As mentioned earlier, the deterministic communication complexity for the
count-tracking problem has been settled at $\Theta(k/\eps\cdot \log N)$
\cite{yi09:_optim}\footnote{The lower bound in \cite{yi09:_optim} was
  stated for the heavy hitters tracking problem, but essentially the same
  proof works for count-tracking.}, with or without two-way communication.
In this paper, we show that with randomization {\em and} two-way
communication, this is reduced to $\Theta(\sqrt{k}/\eps \cdot \log N)$.  We
first in Section~\ref{sec:track-distr-sum} present a randomized algorithm
with this communication cost that, at {\em any one} given time instance,
maintains an $\eps$-approximation of the current $n$ with a constant
probability.  The algorithm is very simple and uses $O(1)$ space at each
site.  It is easy to make the algorithm correct for {\em all} time
instances and boost the probability to $1-\delta$: Since we can use the
same approximate value $\hat{n}$ of $n$ until $n$ grows by a $1+\eps$
factor, it suffices to make the algorithm correct for $O(\log_{1+\eps} N) =
O(1/\eps\cdot \log N)$ time instances.  Then running $O(\log(\frac{\log
  N}{\delta\eps}))$ independent copies of the algorithm and taking the
median will achieve the goal of tracking $n$ continuously at all times,
with probability at least $1-\delta$.  The $\Omega(\sqrt{k}/\eps \cdot \log
N)$ lower bound (Section~\ref{sec:lower-bound}) actually holds on the
number of messages that have to be exchanged, regardless of the message
size, and holds even assuming the sites have unlimited space and computing
power.  That randomization is necessary to achieve this $\sqrt{k}$-factor
improvement follows from the previous deterministic lower bound
\cite{yi09:_optim}; here in Section~\ref{sec:lower-bound} we give an proof
that two-way communication is also required.  More precisely, we show that
any randomized algorithm with one-way communication has to use
$\Omega(k/\eps \cdot \log N)$ communication, i.e., the same as that for
deterministic algorithms.

In the {\em frequency-tracking} (a.k.a.\ {\em heavy hitters tracking})
problem, $A(t)$ is a multiset of cardinality $n(t)$ at time $t$. Let
$f_j(t)$ be the frequency of element $j$ in $A(t)$.  The goal is to
maintain a data structure from which $f_j(t)$, for any given $j$, can be
estimated with absolute error at most $\eps n(t)$, with probability at
least $0.9$ (say).  Note that this problem degenerates to count-tracking
when there is only one element.  It is reasonable to ask for an error in
terms of $n(t)$: if the error were $\eps f_j(t)$, then every element would
have to be reported if they were all distinct.  In fact, this error
requirement is the widely accepted definition for the heavy hitters
problem, which has been extensively studied in the streaming literature
\cite{cormode08:_findin1}.  Several algorithms with the optimal $O(1/\eps)$
space exist \cite{misra:finding,manku02:_approx,metwally06}.  In the
distributed tracking model, we previously \cite{yi09:_optim} gave a
deterministic algorithm with $O(k/\eps\cdot \log N)$ communication, which
is the best possible for deterministic algorithms.  In this paper, by
generalizing our count-tracking algorithm, we reduce the cost to
$O(\sqrt{k}/\eps\cdot \log N)$, with randomization
(Section~\ref{sec:track-distr-freq}).  Since this problem is more general
than count-tracking, by the count-tracking lower bound, this is also
optimal.  Our algorithm uses $O(1/(\eps\sqrt{k}))$ space to process the
stream at each site, which is actually smaller than the $\Omega(1/\eps)$
space lower bound for this problem in the streaming model.  This should not
come at a surprise: Due to the fact that the site is allowed to communicate
to the coordinator {\em during} the streaming process, the streaming lower
bounds do not apply in our model.  To this end, we prove a new space lower
bound of $\Omega(1/(\eps \sqrt{k}))$ bits for our model, showing that our
algorithm also uses near-optimal space.  This space lower bound is
conditioned upon the requirement that the communication cost should be
$O(\sqrt{k}/\eps \cdot\log N)$ bits.  Note that it is not possible to
prove a space lower bound unconditional of communication: A site can send
every element to the coordinator and thus only needs $O(1)$ space.  In
fact, what we prove is a space-communication trade-off; please see
Section~\ref{sec:lower-bound-1} for the precise statement.

For the {\em rank-tracking} problem, it will be convenient to assume that
the elements are drawn from a totally ordered universe and $A(t)$ contains
no duplicates.  The {\em rank} of an element $x$ in $A(t)$ ($x$ may not be
in $A(t)$) is the number of elements in $A(t)$ smaller than $x$, and our
goal is to compute a data structure from which the rank of any given $x$
can be estimated with error at most $\eps n(t)$, with constant probability.
Note that a rank-tracking algorithm also solves the frequency-tracking
problem (but not vice versa), by turning each element $x$ into a pair
$(x,y)$ to break all ties and maintaining such a rank-tracking data
structure.  When the frequency of $x$ is desired, we ask for the ranks of
$(x,0)$ and $(x,\infty)$ and take the difference.  We previously
\cite{yi09:_optim} gave a deterministic algorithm for the rank-tracking
problem with communication $O(k/\eps \cdot \log N \log^2(1/\eps))$.  In
this paper, we show in Section~\ref{sec:rank-estimation} how randomization
can bring this down to $O(\sqrt{k}/\eps \cdot \log N
\log^{1.5}(1/\eps\sqrt{k}))$, which is again optimal ignoring
$\polylog(1/\eps, k)$ factors.  Since rank-tracking is more general than
frequency-tracking, the previous lower bounds also hold here.  Our
algorithm uses space that is also close to the $\Omega(1/(\eps \sqrt{k}))$
lower bound.

Since we are talking about randomized algorithms with a constant success
probability, we should also compare with random sampling. It is well known
\cite{vc-ucrfe-71} that this probabilistic guarantee can be achieved for
all the problems above by taking a random sample of size $O(1/\eps^2)$.  A
random sample can be maintained continuously over distributed streams
\cite{cormode10:_optim}, solving these distributed tracking problems, with
a communication cost of $O(1/\eps^2 \cdot \log N)$.  This is worse than our
algorithms when $k = o(1/\eps^2)$.  As noted earlier, all the upper bounds
we have mentioned above have a hidden additive $O(k\log N)$ term, including
that for the random sampling algorithm.  Thus when $k = \Omega(1/\eps^2)$,
all of them boil down to $O(k \log N)$, while $\Omega(k)$ is an easy lower
bound for all these problems (see Theorem~\ref{th:k-lower}).  This means
that when $k=\Omega(1/\eps^2)$, all problems can be solved optimally by
just random sampling, up to an $O(\log N)$ factor.  Therefore,
$k=o(1/\eps^2)$ is the more interesting case worthy of studying.  In
addition, as the error (in particular for the frequency-tracking and the
rank-tracking problems) is in terms of $n$, the current size of the {\em
  entire} data set, typical values of $\eps$ are quite small.  For example,
$\eps=10^{-2} \sim 10^{-4}$ was used in the experimental study
\cite{cormode08:_findin1} for these problems in the streaming model; while
$k$ usually ranges from $10$ to $10^4$.  Thus we will assume $k\le
1/\eps^2$ in all the upper bounds throughout the paper.

The idea behind all our algorithms is very simple.  Instead of
deterministic algorithms, we use randomized algorithms that produce
unbiased estimators for $n_i$, the frequencies, and ranks with variance
$(\eps n)^2/k$, leading to an overall variance of $(\eps n)^2$, which is
sufficient to produce an estimate within error $\eps n$ with constant
probability.  This means we can afford an error of $\eps n / \sqrt{k}$ from
each site, as opposed to $\eps n / k$ for deterministic algorithms.  This
is essentially where we obtain the $\sqrt{k}$-factor improvement by
randomization.  Our algorithms are simple and extremely lightweight, in
particular the count-tracking and frequency-tracking algorithms, thus can be
easily implemented in power-limited distributed systems like wireless sensor networks.


\subsection{Other related work}
As distributed tracking is closely related to the streaming and the
$k$-party communication model, it could be enlightening to compare with
the known results of the above problems in these models.  As mentioned
earlier, the count-tracking problem is trivial in both models, requiring
$O(1)$ space in the streaming model and $O(k)$ communication in the
$k$-party communication model.  

Both the frequency-tracking and rank-tracking problems have been
extensively studied in the streaming model with a long history.  The former
was first resolved by the MG algorithm \cite{misra:finding} with the
optimal space $O(1/\eps)$, though several other algorithms with the same
space bound have been proposed later on \cite{manku02:_approx,metwally06}.
The rank problem is also one of the earliest problems studied in the
streaming model \cite{munro80:_selec}.  The best deterministic algorithm to
date is the one by Greenwald and Khana \cite{greenwald01:_space}.  It uses
$O(1/\eps \cdot \log n)$ working space to maintain a structure of size
$O(1/\eps)$, from which any rank can be estimated with error $\eps n$.
Note that the rank problem is often studied as the {\em quantiles} problem
in the literature.  Recall that for any $0\le \phi\le1$, the
$\phi$-quantile of $D$ is the element in $A(t)$ that ranks at $\lfloor \phi
n \rfloor$, while an $\eps$-approximate $\phi$-quantile is any element that
ranks between $(\phi-\eps)n$ and $(\phi+\eps)n$.  Clearly, if we have the
data structure for one problem, we can do a binary search to solve the
other.  Thus the two problems are equivalent, for deterministic algorithms.
For algorithms with probabilistic guarantees, we need all $O(\log(1/\eps))$
decisions in the binary search to succeed, which requires the failure
probability to be lowered by an $O(\log(1/\eps))$ factor.  By running
$O(\log\log(1/\eps))$ independent copies of the algorithm, this is not a
problem.  So the two problems differ by at most a factor of
$O(\log\log(1/\eps))$.

The existing streaming algorithms for the frequency and rank problems can
be used to solve the one-shot version of the problem in the $k$-party
communication model easily.  More precisely, we use a streaming algorithm
to summarize the data set at each site with a structure of size
$O(1/\eps)$, and then send the these summary structures to the coordinator,
resulting in a communication cost of $O(k/\eps)$.  Recently, we designed
randomized algorithms for these two problems with $O(\sqrt{k}/\eps)$
communication \cite{huang11,huang11:_optim}, which have just been shown to
be near-optimal in an unpublished manuscript \cite{verbin:_tight}.  Thus, the
results in this paper demonstrate that, the seemingly much more challenging
tracking problem, which requires us to solve the one-shot problem
continuously at all times, is only harder by an $\Theta(\log N)$
factor (except for the count-tracking problem, which is much harder than
its one-shot version). 

Finally, we should mention that all these distributed tracking problems
have been studied in the database community previously, but mostly using
heuristics.  Keralapura et al.~\cite{keralapura06} approached the
count-tracking problem using prediction models, which do not work under
adversarial inputs.  Babcock and Olston \cite{Babcock:Olston:03} studied
the top-$k$ tracking problem, a variant of the frequency (heavy hitters) tracking
problem, but did not offer a theoretical analysis.  The rank-tracking
problem was first studied by Cormode et
al.~\cite{Cormode:Garofalakis:Muthukrishnan:Rastogi:05}; their algorithm
has a communication cost of $O(k/\eps^2 \cdot \log N)$ under certain
inputs.


\section{Tracking Distributed Count}

\subsection{The algorithm}
\label{sec:track-distr-sum}

\paragraph{The algorithm with a fixed $p$}
Let $p$ be a parameter to be determined later.  For now we will assume that
$p$ is fixed. The algorithm is very simple: Whenever site $S_i$ receives an
element (hence $n_i$ gets incremented by one), it sends the latest value of
$n_i$ to the coordinator with probability $p$.  Let $\bar{n}_i$ be the last
updated value of $n_i$ received by the coordinator.  We first
estimate each $n_i$ by
\begin{equation}
\label{eq:1}
 \hat{n}_i  = \left\{\begin{array}{ll}
 \bar{n}_i - 1 + 1/p, & \textrm{if }\bar{n}_i\textrm{ exists}; \\
0, & \textrm{else.}
\end{array}\right.\end{equation}
Then we estimate $n$ as $\hat{n} = \sum_i \hat{n}_i$.

\paragraph{Analysis}
As mentioned in the introduction, our analysis will hold for any given one
time instance.  It is also important to note that this given time instance
shall not depend on the randomization internal to the algorithm.

We show that each $\hat{n}_i$ is an unbiased estimator of $n_i$ with
variance at most $1/p^2$.  This is very intuitive, since $n_i -
\bar{n}_i$ is the number of failed trials until the site decides to send an
update to the coordinator, when we look backward from the current time
instance.  This follows a geometric distribution with parameter $p$, but
not quite, as it is bounded by $n_i$.  This is why we need to separate the
two cases in \eqref{eq:1}.  A more careful analysis is given below:

\begin{lemma} \label{sample}
$\E[\hat{n}_i] = n_i$; $\var[\hat{n}_i] \le 1/p^2$.
\end{lemma}
\begin{proof}
Define the random variable
\[
X = \left\{
\begin{array}{ll}
n_i - \bar{n}_i + 1, & \textrm{if $\bar{n}_i$ exists;}\\
n_i+1/p, & \textrm{else.}
\end{array}\right. \]
Now we can rewrite $\hat{n}_i$ as $\hat{n}_i = n_i - X + 1/p$.  Thus it
suffices to show that $\E[X] = 1/p$ and $\var[X]\le 1/p^2$.  Letting $t=n_i
-\bar{n}_i+1$, we have
\begin{align*}
\E[X] &= \sum_{t=1}^{n_i} (t(1-p)^{t-1} p) + (n_i+1/p)(1-p)^{n_i} = \frac{1}{p}.\\
\var[X] &= \sum_{t=1}^{n_i} ((t - 1/p)^2(1-p)^{t-1} p) + (n_i+1/p - 1/p)^2(1-p)^{n_i}\\
            &= \frac{(1-p)(1-(1-p)^{n_i})}{p^2} \le \frac{1}{p^2}.
\end{align*}
\end{proof}

By Lemma~\ref{sample}, we know that $\hat{n}$ is an unbiased estimator of
$n$ with variance $\le k/p^2$.  Thus, if $p=\sqrt{k}/\eps n$, the variance
of $\hat{n}$ will be $(\eps n)^2$, which means that $\hat{n}$ has error at
most $2\eps n$ with probability at least $3/4$, by Chebyshev inequality.
Rescaling $\eps$ and $p$ by a constant will reduce the error to $\eps n$
and improves the success probability to $0.9$, as desired.  Here we also
see that separating the two cases in \eqref{eq:1} is actually important.
Otherwise, when $n_i = \Theta(\eps n / \sqrt{k})$, there would be a
constant probability that $\bar{n}_i$ does not exist, leading to a bias of
$\Theta(1/p) = \Theta(\eps n / \sqrt{k})$.  Summing over all $k$ sites,
this would exceed our error requirement.

It is interesting to note that similar ideas were used to solve the
{\em one-shot} quantile problem over distributed data \cite{huang11}.

\paragraph{Dealing with a decreasing $p$}
It is not possible and necessary to set $p$ exactly to $\sqrt{k}/\eps n$.
From the analysis above, it should be clear that keeping
$p=\Theta(\sqrt{k}/\eps n)$ will suffice.  To do so, we first track $n$
within a constant factor.  This can be done efficiently as follows.  Each
site $S_i$ keeps track of its own counter $n_i$.  Whenever $n_i$ doubles,
it sends an update to the coordinator. The coordinator sets
$n'=\sum_{i=1}^k n_i'$, where $n_i'$ is the last update of $n_i$.  When
$n'$ doubles (more precisely, when $n'$ changes by a factor between $2$ and
$4$), the coordinator broadcasts $n'$ to all the sites.  Let $\bar{n}$ be
the last broadcast value of $n'$.  It is clear that $\bar{n}$ is always a
constant-factor approximation of $n$.  The communication cost is $O(k\log
N)$, since each site sends $O(\log N)$ updates to the coordinator and the
coordinator broadcasts $O(\log N)$ times, each of which costs $k$ messages.
These broadcasts divide the whole tracking period into $O(\log N)$ rounds,
and within each round, $n$ stays within a constant factor of $\bar{n}$, the
broadcast value at the beginning of the round.

Now, when $\bar{n} \le \sqrt{k}/\eps$, we set $p=1$.  This causes all the
first $O(\sqrt{k}/\eps)$ elements to be sent to the coordinator.  When
$\bar{n} > \sqrt{k}/\eps$, we set $p = 1/\lfloor \eps \bar{n} / \sqrt{k}
\rfloor_2$, where $\lfloor x \rfloor_2$ denotes the largest power of $2$
smaller than $x$.  Since $\bar{n}$ is monotonically increasing, $p$ gets
halved over the rounds.  At the beginning of a round, if the new $p$ is
half\footnote{To be more precise, the new $p$ might also be a quarter of
  the previous $p$, but it can be handled similarly. } of that in the
previous round, each site $S_i$ adjusts its $\bar{n}_i$ appropriately, as
follows.  First with probability $1/2$, the site decides if $\bar{n}_i$
remains the same.  If so, nothing changes; otherwise, it repeatedly flips a
coin with probability $1/p$ (with the new $p$).  Every failed coin flip
decrements $\bar{n}_i$ by one.  It does so until a successful coin flip, or
$\bar{n}_i=0$.  Finally, the site informs the coordinator of the new value
of $\bar{n}_i$; if $\bar{n}_i=0$, the coordinator will treat it as if
$\bar{n}_i$ does not exist.  It should be clear that after this adjustment,
the whole system looks as if it had always been running with the new $p$.

It is easy to see that the communication cost in each round is $O(k + p n)
= O(k+\sqrt{k}/\eps) = O(\sqrt{k}/\eps)$, thus the total cost is
$O(\sqrt{k}/\eps\cdot \log N)$.

\begin{theorem}
\label{thm:sum}
  There is an algorithm for the count-tracking problem that, at any time,
  estimates $n=\sum_i n_i$ within error $\eps n$ with probability at least $0.9$.
  It uses $O(1)$ space at each site and $O(\sqrt{k}/\eps \cdot \log N)$
  total communication.
\end{theorem}

\subsection{The lower bound} 
\label{sec:lower-bound}
Before proving the lower bounds, we first state our lower bound model
formally, in the context of the count-tracking problem.  The $N$ elements
arrive at the $k$ sites in an online fashion at arbitrary time instances.
We do not allow spontaneous communication.  More precisely, it means that a
site is allowed to send out a message only if it has just received an
element or a message from the coordinator.  Likewise, the coordinator is
allowed to send out messages only if it has just received messages from one
or more sites.  When a site $S_j$ is allowed to send out a message, it
decides whether it will indeed do so and the content of the message, based
only on its local counter $n_j$ and the message history between $S_j$ and
the coordinator, possibly using some random source.  We assume that the
site does not look at the current clock.  We argue that the clock conveys
no information since the elements arrive at arbitrary and unpredictable
time instances.  (If the elements arrive in a predictable fashion, say, one
per time step, the problem can be solved without communication al all.)
Similarly, when the coordinator is allowed to send out messages, it makes
the decision on where and what to send based only on its message history
and some random source.  We will lower bound the communication cost only by
the number of messages, regardless of the message size.


\subsubsection{One-way communication lower bound}
In this section we show that two-way communication is necessary to achieve
the upper bound in Theorem~\ref{thm:sum}, by proving the following lower bound.  Remember that we assume $N$ is sufficiently larger than $k$ and $1/\eps$.

\begin{theorem}
\label{thm:one-way}
If only the sites can send messages to the coordinator but not vice versa,
then any randomized algorithm for the count-tracking problem that, at any
time, estimates $n$ within error $\eps n$ with probability at least $0.9$
must send $\Omega(k/\eps\cdot\log N)$ messages. 
\end{theorem}

\begin{proof}
We first define the hard input distribution $\mu$.
\begin{enumerate}
\item[(a)] With probability $1/2$, all elements arrive at one site that is uniformly picked at random.
\item[(b)] Otherwise, the $N$ elements arrive at the $k$ sites in a round-robin fashion, each site receiving $N/k$ elements in the end.
\end{enumerate}
By Yao's Minimax principle~\cite{yao77}, we only need to argue that any
deterministic algorithm with success probability at least $0.8$ under $\mu$ has expected cost $\Omega(k/\eps \cdot \log N)$.

Note that when only one-way communication is allowed, a site decides whether
to send messages to the coordinator only based on its local counter $n_j$.  Thus the communication pattern can be essentially described as follows.  Each site $S_j$ has a series of thresholds $t_j^1, t_j^2, \ldots$ such that when
$n_j = t_j^i$, the site sends the $i$-th message to the
coordinator. These thresholds should be fixed at the beginning.

We lower bound the communication cost by rounds.  Let $W_i$ be the number of elements that have arrived up until round $i$.  We divide the rounds by setting $W_1=k/\eps$, and $W_{i+1} = \lceil (1+\eps) W_i \rceil$ for $i\ge 1$.  Thus there are $1/\eps \cdot \log(\eps N / k)$ rounds, which is $\Omega(1/\eps \cdot \log N)$ for sufficiently large $N$.

At the beginning of round $i+1$, suppose that $S_1, S_2, \ldots, S_k$ have
already sent $z_1^i, z_2^i, \ldots, z_k^i$ messages to the coordinator,
respectively. Let $t_{\max}^{i+1} = (1 + \eps) \cdot \max\{t_j^{z_j^i}\ |\ 
j = 1, 2, \ldots, k\}$.  We first observe that there must be at least $k/2$
sites with their next threshold $t_j^{z_j^i+1} \le t_{\max}^{i+1}$.
Otherwise, suppose there are less than $k/2$ sites with such next
thresholds, then with probability at least $1/4$ case (a) happens and the
random site $S_j$ chosen to receive all elements has $t_j^{z_j^i+1} >
t_{\max}^{i+1}\ge (1+\eps) t_j^{z_j^i}$.  Thus, with probability at least
$1/4$ the algorithm fails when the $t_{\max}^{i+1}$-th element arrives,
contradicting the success guarantee.

On the other hand, with probability $1/2$ case (b)
happens. In this case all $t_j^{z_j^i}\ (j = 1, 2, \ldots, k)$ are
no more than $W_i/k$, since in case (b), elements arrive at all $k$ sites in
turn. In the next $\eps W_i$ elements, each site $S_j$ receives $\eps W_i/k$ elements.  If the site $S_j$ has 
$t_j^{z_j^i+1} \le t_{\max}^{i+1}$, then it must send a message in this round, since $W_i/k + \eps W_i/k \ge t_{\max}^{i+1} \ge t_j^{z_j^i+1}$, that is, its $(z_j^i+1)$-th threshold is triggered.  As argued, there are $\ge k/2$ sites with $t_j^{z_j^i+1} \le t_{\max}^{i+1}$, so the
communication cost in this round is at least $k/2$.

Summing up all rounds,  the total communication is at least $\Omega(k/\eps
\cdot \log N)$.
\end{proof}

\subsubsection{Two-way communication lower bound}
Below we prove two randomized lower bounds when two-way communication is
allowed.  The first one justifies the assumption $k\le 1/\eps^2$, since
otherwise, random sampling will be near-optimal.

\begin{theorem}
\label{th:k-lower}
Any randomized algorithm for the count-tracking problem that, at any time,
estimates $n$ within error $0.1n$ with probability at least $0.9$ must
exchange $\Omega(k)$ messages.
\end{theorem}

\begin{proof}
The hard input distribution is the same as that in the proof of Theorem~\ref{thm:one-way}. To prove this lower bound we are only interested in the number of sites that communicate with the coordinator at least once. Before any element arrives, we can still assume that each site keeps a triggering threshold. The thresholds of $S_j$ shall remain the same unless it communicates with the coordinator at least once. We argue that there must be at least $k/2$ sites whose triggering threshold is no more than $1$, since otherwise if case (a) happens and the randomly chosen site is one with a triggering threshold larger than $1$, the algorithm will fail, which would happen with probability at least $1/4$.  On the other hand, if case (b) happens, then all the sites with threshold $1$ will
have to communicate with the coordinator at least once: either their thresholds are triggered by the round-robin arrival of elements, or they receive a message from the coordinator, which can possibly change their threshold.
\end{proof}

Finally, we show that the upper bound in Theorem~\ref{thm:sum} is
asymptotically tight. We first introduce the following primitive problem.

\begin{definition}[$1$-bit]
Let $s$ be either $k/2 + \sqrt{k}$ or $k/2 - \sqrt{k}$, each with probability $1/2$. 
From the $k$ sites, a subset of $s$
sites picked uniformly at random each have 
bit $1$, while the other $k - s$ sites have bit $0$. The goal
of the communication problem is for the coordinator to find out the
value of $s$ with probability at least $0.8$.
\end{definition}
We will show the following lower bound for this primitive problem.
\begin{lemma}
\label{lem:1bit}
Any deterministic algorithm that solves $1$-bit has distributional
communication complexity $\Omega(k)$.
\end{lemma}

Lemma~\ref{lem:1bit} immediately implies the following theorem:

\begin{theorem}
  Any randomized algorithm for the count-tracking problem that, at any
  time, estimates $n$ within error $\eps n$ with probability at least $0.9$
  must exchange $\Omega(\sqrt{k}/\eps\cdot \log N)$ messages, when
  $k<1/\eps^2$.
\end{theorem}
\begin{proof}
We will again fix a hard input distribution first and then focus on the
distributional communication complexity of deterministic algorithms with success probability at most $0.8$. Let
$[m] = \{0, 1, \ldots, m-1\}$. The adversarial input consists of $\ell
= \log\frac{\eps N}{k} = \Omega(\log N)$ rounds. We further divide each round $i \in [\ell]$
into $r = 1/(2\eps\sqrt{k})$ subrounds.

The input at round $i \in [\ell]$ is constructed as follows, at each
subround $j \in [r]$, we first choose $s$ to be $k/2 + \sqrt{k}$ or $k/2 -
\sqrt{k}$ with equal probability. Then we choose $s$ sites out of the
$k$ sites uniformly at random and send $2^i$ elements to each of them (the order does not matter).


It is easy to see that at the end of in each subround in round $i$, the
total number of items is no more than $\tau_i = \sqrt{k}/\eps \cdot 2^i$.
Thus after $s \cdot 2^i$ elements have arrived in a subround, the algorithm
has to correctly identify the value of $s$ with probability at least $0.8$,
since otherwise with probability at least $0.2$ the estimation of the
algorithm will deviate from the true value by at least $\sqrt{k} \cdot 2^i
> \eps \tau_i$, violating the success guarantee of the algorithm.  This is
exactly the 1-bit problem defined above.  By Lemma~\ref{lem:1bit}, the
communication cost of each subround is $\Omega(k)$.  Summing over all $r$
subrounds and then all $\ell$ rounds, we have that the total communication
is at least $\ell \cdot r \cdot \Omega(k) \ge \Omega(\sqrt{k}/\eps \cdot
\log N)$.
\end{proof}

Now we prove Lemma~\ref{lem:1bit}.
\begin{proof}(of Lemma~\ref{lem:1bit})
First of all, observe that whenever the coordinator communicates with
a site, the site can send its whole input (i.e., its only bit) to the
coordinator. After that, the coordinator knows all the information about
that site and does not need to communicate with it further. Therefore
all that we need to investigate is the number of sites the coordinator
needs to communicate with.

There can be two types of actions in the protocol.
\begin{enumerate}
\item[(a)] A site initiates a communication with the coordinator based on
  the bit it has.
\item[(b)] The coordinator, based on all the information it has gathered so far, asks some site to send its bit.
\end{enumerate}

Note that if a type (b) communication takes place before a type (a)
communication, we can always swap the two, since this only gives the
coordinator more information at an earlier stage.  Thus we can assume that
all the type (a) communications happen before type (b) ones.

In the first phase where all the type (a) communications happen, let $x$
be the number of sites that send bit $0$ to the coordinator, and $y$ be the
number of sites that send bit $1$ to the coordinator. If $\E[x+y] =
\Omega(k)$, then we are done. So let us assume that $\E[x+y] = o(k)$. By
Markov inequality we have that, with probability at least $0.9$, $x+y =
o(k)$. After the first phase, the problem becomes that there are $s' = s -
y = s - o(k)$ sites having bit $1$, out of a total $k' = k - x - y = k -
o(k)$ sites. The coordinator needs to figure out the exact value of $s'$
with probability at least $0.8- (1 - 0.9) = 0.7$.

In the second phase where all type (b) communication happens, from the
coordinator's perspective, all the remaining sites are still symmetric (by
the random input we choose), therefore the best it can do is to probe an
arbitrary site among those that it has not communicated with. This is still
true even after the coordinator has probed some of the remaining sites.
Therefore, the problem boils down to the following: The coordinator picks
$z$ sites out of the remaining $k'$ sites to communicate and then decides
the value of $s'$ with success probability at least $0.7$.  We call this
problem the {\em sampling} problem.  We can show that to achieve the
success guarantee, $z$ should be at least $\Omega(k)$.  This result is
perhaps folklore; proofs to more general versions of this problem can be
found in \cite{Bar-Yossef:02} (Chapter 4), and also \cite{patt:08,woodruff:07}.
We include a simpler proof in the appendix for completeness. With this we
conclude the proof of Lemma~\ref{lem:1bit}.
\end{proof}

\section{Tracking Distributed Frequencies}
\label{sec:track-distr-freq}
In the frequency-tracking problem, $A$ (we omit ``$(t)$'' when the context
is clear) is a multiset and the goal is to track the frequency of any item
$j$ within error $\eps n$. Let $f_{ij}$ denote the local frequency of
element $j$ in $A_i$, and let $f_j = \sum_{i=1}^k f_{ij}$.

\subsection{The algorithm}
\paragraph{The algorithm with a fixed $p$}
As in Section~\ref{sec:track-distr-sum} we first describe the algorithm
with a fixed parameter $p$.  If each site tracks the local frequencies
$f_{ij}$ exactly, we can essentially use the count-tracking algorithm to
track the $f_j$'s.  To achieve small space, we make use of the following
algorithm due to Manku and Motwani \cite{manku02:_approx} at each site
$S_i$: We maintain a list $L_i$ of counters. When an element $j$ arrives
at $S_i$, it first checks if there is a counter $c_{ij}$ for $j$ in $L_i$.
If yes, we increase $c_{ij}$ by $1$.  Otherwise, we sample this element
with probability $p$.  If it is sampled, we insert a counter $c_{ij}$,
initialized to $1$, into $L_i$.  It is easy to see that the expected size
of $L_i$ is $O(pn_i)$.

Next, we follow a similar strategy as in the count-tracking algorithm: The
site reports the counter $c_{ij}$ to the coordinator when it is first added
to the counter list with an initial value of $1$.  Afterward, for every $j$
that is arriving, the site always increments $c_{ij}$ as before, but only
sends the updated counter to the coordinator with probability $p$.  We
use $\bar{c}_{ij}$ to denote the last updated value of $c_{ij}$.

The tricky part is how the coordinator estimates $f_{ij}$, hence $f_j$.
Fix any time instance.  The difference between $f_{ij}$ and $\hat{c}_{ij}$
comes from two sources: one is the number of $j$'s missed before a copy is
sampled, and the other is the number of $j$'s that arrive after the last
update of $c_{ij}$.  It is easy to see that both errors follow the same
distribution as $n_i - \bar{n}_i$ in the count-tracking algorithm.  Thus it
is tempting to modify \eqref{eq:1} as
\begin{equation}
\label{eq:4}
\hat{f}_{ij} = \left\{
\begin{array}{ll}
\bar{c}_{ij} -2 + 2/p, & \textrm{if $\bar{c}_{ij}$ exists;}\\
0, & \textrm{else.}
\end{array}\right.
\end{equation}
However, this estimator is biased and its bias might be as large as
$\Theta(\eps n /\sqrt{k})$.  Summing over $k$ streams, this would exceed
our error guarantee.  To see this, consider the $f_{ij}$ copies of $j$.
Effectively, the site samples every copy with probability $p$, while
$\bar{c}_{ij}-2$ is exactly the number of copies between the first and the
last sampled copy (excluding both).  We define $X_1$ as before
\begin{displaymath}
X_1 = \left\{
\begin{array}{ll}
t_1, & \textrm{if the $t_1$th copy is the first one sampled;}\\
f_{ij}+1/p, & \textrm{if none is sampled.}
\end{array}\right.
\end{displaymath}
We define $X_2$ in exactly the same way, except that we examine these
$f_{ij}$ copies backward:
\begin{displaymath}
X_2 = \left\{
\begin{array}{ll}
t_2, & \textrm{if the $t_2$th copy is the first one sampled} \\ 
& \textrm{in the reverse order;}\\ 
f_{ij}+1/p, & \textrm{if none is sampled.}
\end{array}\right.
\end{displaymath}

It is clear that $X_1$ and $X_2$ have the same distribution with $\E[X_1] =
\E[X_2] = 1/p$ (by Lemma~\ref{sample}), so $\hat{f}_{ij} =
f_{ij}-(X_1+X_2)+2/p$ is unbiased.  Since $\bar{c}_{ij}-2 = f_{ij} -
t_1-t_2$, the correct unbiased estimator should be
\begin{equation}
\label{eq:2}
\hat{f}_{ij} = \left\{
\begin{array}{ll}
\bar{c}_{ij} - 2 + 2/p, & \textrm{if $\bar{c}_{ij}$ exists;}\\
-f_{ij}, & \textrm{else.}
\end{array}\right.
\end{equation}

Compared with the previous wrong estimator \eqref{eq:4}, the main
difference is how the estimation is done when no copy of $j$ is sampled.
When $f_{ij} = \Theta(\eps n /\sqrt{k})$ and $p = \Theta(1/f_{ij})$, this
happens with constant probability, which would result in a bias of
$\Theta(f_{ij}) = \Theta(\eps n /\sqrt{k})$.

However, the correct estimator \eqref{eq:2} depends on $f_{ij}$, the
quantity we want to estimate in the first place.  The workaround is to use
another unbiased estimator for $f_{ij}$ when $\bar{c}_{ij}$ is not yet
available.  It turns out that we can just use simple random sampling: The
site samples every element with probability $p$ (this is independent of the
sampling process that maintains the list $L_i$), and sends the sampled
elements to the coordinator.  Let $d_{ij}$ be the number of sampled copies
of $j$ received by the coordinator from site $i$, the final estimator for
$f_{ij}$ is
\begin{equation}
\label{eq:3}
\hat{f}'_{ij} = \left\{
\begin{array}{ll}
  \bar{c}_{ij} - 2 + 2/p, & \textrm{if $\bar{c}_{ij}$ exists;}\\
-d_{ij}/p, & \textrm{else.}
\end{array}\right.
\end{equation}
Since $d_{ij}$ is independent of $\bar{c}_{ij}$, the estimator is still
unbiased.  Below we analyze its variance.

\paragraph{Analysis}
Intuitively, the variance is not affected by using the simple random
sampling estimator $d_{ij}/p$, because it is only used when $\bar{c}_{ij}$
is not available, which means that $f_{ij}$ is likely to be small, and when
$f_{ij}$ is small, $d_{ij}/p$ actually has a small variance.  When $f_{ij}$
is large, $d_{ij}/p$ has a large variance, but we will use it only with small
probability. Below we give a formal proof.

\begin{lemma}
$\E[\hat{f}'_{ij}] = f_{ij}$; $\var[\hat{f}'_{ij}] = O(1/p^2)$.
\end{lemma}
\begin{proof}
  We first analyze the estimator $\hat{f}_{ij}$ of \eqref{eq:2}.  That
  $\E[\hat{f}_{ij}] = f_{ij}$ follows from the discussion above.  Its
  variance is $\var[\hat{f}_{ij}] = \var[X_1+X_2]$.  Note that $X_1$ and
  $X_2$ are not independent, but they both have expectation $1/p$ and
  variance $\le 1/p^2$.  We first rewrite
\begin{eqnarray*}
\var[X_1+X_2]&=&\E[X_1^2+X_2^2+2X_1X_2]-\E[X_1+X_2]^2 \\
  &=& \var[X_1] + \E[X_1]^2 + \var[X_2] + \E[X_2]^2 \\
&&+ 2\E[X_1X_2] -
  (\E[X_1] + \E[X_2])^2 \\
&\le& 4/p^2 + 2\E[X_1X_2] - 4/p^2 \le 2\E[X_1X_2].
\end{eqnarray*}
Let $\mathcal{E}_t$ be the event that the $t$th copy of $j$ is the first
being sampled.  We have
\begin{eqnarray*}
&&\E[X_1X_2]\\
 &=&\sum_{t=1}^{f_{ij}}(1-p)^{t-1} pt\E[X_2 \mid \mathcal{E}_t] +
       (1-p)^{f_{ij}}(f_{ij} + 1/p)^2\\
       &=& \sum_{t=1}^{f_{ij}}(1-p)^{t-1}
      pt\left((1-p)^{f_{ij}-t}(f_{ij}-t+1)+\sum_{l=1}^{f_{ij}-t}(1-p)^{l-1}pl\right) \\
&&+ (1-p)^{f_{ij}}(f_{ij} + 1/p)^2\\
   &\le&     \frac{1}{p^2}+(1-p)^{f_{ij}}f_{ij}^2+\frac{(1-p)^{f_{ij}}f_{ij}}{p}.
\end{eqnarray*}
Let $c=f_{ij}p$.  If $c\le 2$, $f_{ij} \le 2/p$, and the variance is
$O(1/p^2)$. Otherwise
$$\E[X_1X_2]\le \frac{1}{p^2}+\frac{c^2}{p^2e^c} + \frac{c}{p^2 e^c}=
O(1/p^2),$$
since $c^2\le e^c$ when $c>2$.

Next we analyze the final estimator $\hat{f}'_{ij}$ of \eqref{eq:3}.
First, $d_{ij}$ is the sum of $f_{ij}$ Bernoulli random variables with
probability $p$, so $\E[d_{ij}/p]= f_{ij}$ and $\var[d_{ij}/p] \le
f_{ij}p/p^2 = f_{ij}/p$. Let $\mathcal{E}_*$ be the event that
$\hat{c}_{ij}$ is available, i.e., at least one copy of $j$ is sampled, and
$\mathcal{E}_0 = \overline{\mathcal{E}_*}$, then
\begin{eqnarray*}
\E[\hat{f}'_{ij}] &=& \E[\hat{f}_{ij} \mid
                 \mathcal{E}_*]\Pr[\mathcal{E}_*]+\E[-d_{ij}/p \mid
                 \mathcal{E}_0]\Pr[\mathcal{E}_0]\\
                 &=& \E[\hat{f}_{ij} \mid \mathcal{E}_*]\Pr[\mathcal{E}_*]+
                 (-f_{ij})\Pr[\mathcal{E}_0]\\
                 &=& \E[\hat{f}_{ij}] = f_{ij}.
\end{eqnarray*}
The variance is
\begin{eqnarray*}
\var[\hat{f}'_{ij}] &=& \E[\hat{f}'^2_{ij}] - \E[\hat{f}'_{ij}]^2\\
&=& \E[\hat{f}_{ij}^2 \mid
\mathcal{E}_*]\Pr[\mathcal{E}_*]+\E[(d_{ij}/p)^2 \mid
\mathcal{E}_0]\Pr[\mathcal{E}_0] - f_{ij}^2\\
                 &=& \E[\hat{f}_{ij}^2 \mid
                 \mathcal{E}_*]\Pr[\mathcal{E}_*]- f_{ij}^2+
                 \E[(d_{ij}/p)^2]\Pr[\mathcal{E}_0]\\
                 &=&\E[\hat{f}_{ij}^2 \mid
                 \mathcal{E}_*]\Pr[\mathcal{E}_*]- f_{ij}^2+
                 (\var[d_{ij}/p]+ f_{ij}^2)\Pr[\mathcal{E}_0]
\end{eqnarray*}
Note that 
\begin{eqnarray*}
\var[\hat{f}_{ij}]& = &\E[\hat{f}_{ij}^2] -f_{ij}^2\\
& =& \E[\hat{f}_{ij}^2 \mid \mathcal{E}_*] \Pr[\mathcal{E}_*] +
\E[\hat{f}_{ij}^2 \mid \mathcal{E}_0] \Pr[\mathcal{E}_0] - f_{ij}^2\\
&=& \E[\hat{f}_{ij}^2 \mid \mathcal{E}_*] \Pr[\mathcal{E}_*] +
f_{ij}^2 \Pr[\mathcal{E}_0] - f_{ij}^2,
\end{eqnarray*}
so
\begin{eqnarray*}
\var[\hat{f}'_{ij}]  &=& \var[\hat{f}_{ij}]+ \var[d_{ij}/p] \Pr[\mathcal{E}_0]\\
               &  \le &\var[\hat{f}_{ij}]+
                 \frac{f_{ij}}{p}\cdot(1-p)^{f_{ij}}.
\end{eqnarray*}
Due to the same reason as above, the second term is $O(1/p^2)$, and the
proof completes.
\end{proof}

\paragraph{Dealing with a decreasing $p$}
As in the count-tracking algorithm, we divide the whole tracking period
into $O(\log N)$ rounds.  Within each round, $n$ stays within a constant
factor of $\bar{n}$, while $\bar{n}$ remains fixed for the whole round.

Within a round, we set the parameter $p$ for all sites to be $p=1/\lfloor
\eps \bar{n} / \sqrt{k} \rfloor_2$.  When we proceed to a new round, all
sites clear their memory and we start a new copy of the algorithm from
scratch with the new $p$.  Given an item $j$, the coordinator estimates its
frequency from each round separately, and add them up.  Since the variance
in a round is $O(k/p^2)$ and $p$ increases geometrically over the rounds,
the total variance is asymptotically bounded by the variance of the last
round, i.e., $O(1/\eps^2)$, as desired.

The space used at some site could still be large, since the site may
receive too many elements in a round.  If all the $O(n)$ elements in a
round have gone to the same site, the site will need to use space $O(pn) =
O(\sqrt{k}/\eps)$.  To bound the space, we restrict the amount of space
used by each site.  More precisely, when a site receives more than
$\bar{n}/k$ elements, it sends a message to the coordinator for
notification, clears its memory, and starts a new copy of the algorithm
from scratch. The coordinator will treat the new copy as if it were a new
site, while the original site no longer receives more elements.  Now the
space used at each site is at most $p\bar{n}/k = O(1/(\eps \sqrt{k}))$.  Since there
are at most $O(k)$ such new ``virtual'' sites ever created in a round, this
does not affect the variance by more than a constant factor.

It remains to show that the total communication cost is $O(\sqrt{k}/\eps
\cdot \log N)$.  From earlier we know that there are $O(\log N)$ rounds;
within each round, $\bar{n}$ is the same and $n$ stays within
$\Theta(\bar{n})$.  Focus on one round.  For each arriving element, the
site $S_i$ updates $\bar{c}_{ij}$ with probability $p$ and also
independently samples it with probability $p$ to maintain $d_{ij}$.  This
costs $O(n \cdot p) = O(\sqrt{k}/\eps)$ communication.

\begin{theorem}
\label{thm:hh}
  There is an algorithm for the frequency-tracking problem that, at any
  time, estimates the frequency of any element within error $\eps n$ with
  probability at least $0.9$.  It uses $O(1/(\eps \sqrt{k})$ space at each site and
  $O(\sqrt{k}/\eps \cdot \log N)$ communication.
\end{theorem}

\subsection{Space lower bound}
\label{sec:lower-bound-1}
It is easy to see that the communication lower bounds for the
count-tracking problem also hold for the frequency-tracking problem. In
this section, we prove the following space-communication trade-off. 

\begin{theorem}
\label{thm:hhlb}
Consider any randomized algorithm for the frequency-tracking problem that, at any
time, estimates the frequency of any element within error $\eps n$ with
probability at least $0.9$.  If the algorithm uses $C$ bits of
communication and uses $M$ bits of space per site, then we must have
$C\cdot M=\Omega(\log N / \eps^2)$, assuming $k\le 1/\eps^2$.
\end{theorem}

Thus, if the communication cost is $C = O(\sqrt{k}/\eps \cdot \log N)$
bits, the space required per site is at least $\Omega(1/(\eps \sqrt{k}))$
bits, as claimed in Table~\ref{tab:results}.  Note that, however, our
algorithm of the previous section uses $O(\sqrt{k}/\eps \cdot \log N)$ {\em
  words} of communication and $O(1/(\eps \sqrt{k}))$ {\em words} of space,
so there is still a small gap between the lower and upper bound.
Interestingly, this lower bound also shows that the random sampling
algorithm \cite{cormode10:_optim} (see Table~\ref{tab:results}) actually
attains the other end of this space-communication trade-off (ignoring the
word/bit difference).

\begin{proof}(of Theorem~\ref{thm:hhlb})
  We will use a result in \cite{verbin:_tight} which states that, under the
  $k$-party communication model, there is an input distribution $\mu_k$
  such that, any algorithm that solves the one-shot version of the problem
  under $\mu_k$ with error $2\eps n$ with probability $0.9$ needs at least $c
  \sqrt{k}/\eps$ bits of communication for some constant $c$, assuming
  $k\le 1/\eps^2$.  Moreover, any algorithm that solves $\ell$ independent
  copies of the one-shot version of the problem needs at least $\ell \cdot
  c \sqrt{k}/\eps$ bits of communication.
  
  We will consider the problem over $\rho k$ sites, for some integer $\rho
  \ge 1$ to be determined later.  We divide the whole tracking period into
  $\log N$ rounds. In each round $i=1,\dots,\log N$, we generate an input
  independently chosen from distribution $\mu_{\rho k}$ to the sites.  We
  pick elements from a different domain for every round so that we have
  $\log N$ independent instances of the problem.  In round $i$, for every
  element $e$ picked from $\mu_{\rho k}$ for any site, we replace it with
  $2^{i-1}$ copies of $e$.  We arrange the element arrivals in a round so
  that site $S_1$ gets all its elements first, then $S_2$ gets all its
  elements, and so on so forth. We will only require the continuous
  tracking algorithm to solve the frequency estimation problem at the end
  of each round.  Since the last round always contains half of all the
  elements that have arrived so far, the algorithm must solve the problem
  for the elements in each round, namely, $\log N$ independent instances of
  the one-shot problem.  By the result in \cite{verbin:_tight}, the
  communication cost to solve all these instances of the problem is at
  least $c \sqrt{\rho k}/\eps \cdot \log N$.
  
  Let $\mathcal{A}_k$ be a continuous tracking algorithm over $k$ sites
  that communicates $C$ bits in total and uses $M$ bits of space per site.
  Below we show how to solve the problem over the $\rho k$ sites in each
  round, by simulating the $k$-site algorithm $\mathcal{A}_k$.  In each
  round, we start the simulation with sites $S_1, \dots, S_k$.  Whenever
  $\mathcal{A}_k$ exchanges a message, we do the same.  When $S_1$ has
  received all its elements, it sends its memory content to $S_{k+1}$,
  which then takes the role of $S_1$ in the simulation and continue.
  Similarly, when $S_2$ has received all its elements, it sends its memory
  content to $S_{k+2}$, which replaces $S_2$ in the simulation.  In
  general, when $S_j$ is done with all its elements, it passes its role to
  $S_{j+k}$.  When $S_{\rho k}$ is done, the simulation finishes for this
  round.  $S_{\rho k}$ then sends a broadcast message and we proceed to the
  next round.
  
  
  Let us analyze the communication cost of the simulation.  First, we
  exchange exact the same messages as $\mathcal{A}_k$ does, which costs
  $C$.  We also communicate $\rho (k-1)$ memory snapshots and a broadcast
  message in each round, which costs $\le \rho k M \log N$ over all rounds.
  Thus, we have
\[ C + \rho k M \log N \ge c\sqrt{\rho k}/\eps \cdot \log N.
\]
Rearranging,
\[
M \ge \frac{c}{\eps\sqrt{\rho k}} - \frac{C}{\rho k \log N} =
\frac{1}{\sqrt{\rho k}}\left(\frac{c}{\eps} - \frac{C}{\sqrt{\rho k}\log N} \right)
\]

Thus, if we set $\sqrt{\rho} = \left\lceil \frac{2C\eps}{c\sqrt{k}\log N}
\right\rceil$, then
\[ M \ge \frac{c}{2\eps\sqrt{\rho k}} = \Omega\left(\frac{\log N}{C\eps^2}\right),
\]
as claimed.
\end{proof}




\section{Tracking Distributed Ranks}
\label{sec:rank-estimation}
On a stream of $n$ elements, an algorithm that produces an unbiased
estimator for any rank with variance $O((\eps n)^2)$ was presented in
\cite{suri06:_range}, which has been very recently improved and made to
work in a stronger model \cite{agarwal:mergeable}.  It uses $O(1/\eps \cdot
\log^{1.5} (1/\eps))$ working space to maintain a rank estimation summary
structure of size $O(1/\eps)$.  We call this algorithm $\A$ and will use it
as a black box in our distributed tracking algorithm.

\paragraph{The overall algorithm}
As before, with $O(k \log N)$ communication, we first track $\bar{n}$, a
constant factor approximation of the current $n$.  This also divides the
tracking period into $O(\log N)$ rounds.  The $\Theta(n)$ elements arriving
in a round are divided into chunks of size at most $\bar{n}/k$, each processed by
an instance of algorithm $\C$, described below.  A site may receive more
than $\bar{n}/k$ elements.  When the $(\bar{n}/k+1)$th element arrives, the site
finishes the current instance of $\C$, and starts a new one, which will
process the next $\bar{n}/k$ elements, and so on so forth.

\paragraph{Algorithm $\C$}
Algorithm $\C$ reads at most $\bar{n}/k$ elements, and divides them into blocks
of size $b=\eps \bar{n}/\sqrt{k}$, so there are at most $\frac{1}{\epsilon
  \sqrt{k}}$ blocks. We build a balanced binary tree on the blocks in the
arrival order, and the height of the tree is $h \le \log \frac{1}{\epsilon
  \sqrt{k}}$. For each node $v$ in the tree, let $D(v)$ be all the elements
contained in the leaves in the subtree rooted at $v$.  For each $D(v)$, we
start an instance of $\A$, denoted as $\A_v$, to process its elements as they
arrive.  We say that $v$ is {\em active} if $\A_v$ is still accepting
elements.  For a node $v$ at level $\ell$ (the leaves are said to be on level
$0$), the error parameter of $\A_v$ is set to $2^{-\ell}/\sqrt{h}$. We say
$v$ is {\em full} if all the elements in $D(v)$ have arrived. When $v$ is
full, we send the summary computed by $\A_v$ to the coordinator, and free
the space used by $\A_v$. Furthermore, for each element that is arriving, we
sample it with probability $p=\frac{\sqrt{k}}{\eps \bar{n}}$, and if it is
sampled, we send it to the coordinator.

\paragraph{Analysis of costs}
We first analyze the various costs of $\C$.  At any time there are at most
$h$ active nodes, one at each level, so the space used by $\C$ is at most
$$\sum_{\ell=0}^{h} \sqrt{h}2^\ell \log^{1.5}\frac{1}{\eps} =
O\left(\frac{\sqrt{h}}{\eps \sqrt{k}} \log^{1.5}\frac{1}{\eps}\right).$$

The communication for $\C$ includes all the summaries computed, and the
elements sampled.   For each $\ell$, the total size of the summaries on level
$\ell$ is
$$O\left(\frac{1}{\eps \sqrt{k}} 2^{-\ell} \cdot
  2^{\ell}\sqrt{h}\right)=O\left(\frac{\sqrt{h}}{\eps \sqrt{k}}\right).$$
Summing over all $h$ levels, it is $\frac{h^{1.5}}{\eps \sqrt{k}}$. There are at
most $2k$ instances of $\C$ in a round, therefore the total communication
cost in a round is $O(h^{1.5}\sqrt{k}/\eps)$.  The number of sampled elements
in a round is $O(n p) = O(\sqrt{k}/\eps)$.  Thus, over all $O(\log N)$
  rounds, the total communication cost is $O(h^{1.5}\sqrt{k}/\eps \cdot
  \log N)$.

\paragraph{Estimation} It remains to show how the coordinator estimates
the rank of any given element $x$ at any time with variance $O((\eps
n)^2)$.  We decompose all $n$ elements that have arrived so far into
smaller subsets, and estimate the rank of $x$ in each of the subsets.
Since all estimators are unbiased, the overall estimator is also unbiased;
the variance will be the sum of all the variances.

We will focus on the current round; all previous rounds can be handled
similarly.  Recall that there are $O(\bar{n})$ elements arriving in this
round and $\bar{n} = \Theta(n)$.  Every chunk of $\bar{n}/k$ elements are
processed by one instance of $\C$.  Consider any such chunk.  Suppose up to
now, $n'$ elements in this chunk have arrived for some $n'\le \bar{n}/k$.
We write $n'$ as $n' = q\cdot b + r$ for some $r<b$, and decompose these
$n'$ elements into at most $h+1$ subsets. The first $qb$ elements are
decomposed into at most $h$ subsets, each of which corresponds to a full
node in the binary tree of $\C$.  The node has already sent its summary to
the coordinator, which we can use to estimate the rank.  For a node at
level $\ell$, the variance is $(2^{-i}/\sqrt{h}\cdot 2^{i} b)^2 = b^2/h$,
so the total variance from all $h$ nodes is $b^2$.

For the last $r$ elements of the chunk that are still being processed by an
active node, the coordinator does not have any summary for them.  But
recall that the site always samples each element with probability $p=\sqrt{k}/(\eps
\bar{n})$ and sends it to the coordinator if it is sampled.  Thus, the
rank of $x$ in these $r$ elements can be estimated by simply counting the
number $c$ of elements sampled that are smaller than $x$, and the estimator
is $c/p$.  The variance of this estimator is $r/p \le b/p = b^2$.  Thus,
the variance from any chunk is $O(b^2)$.  Since there are at most $2k$
chunks in the round, the total variance is $O(b^2 k) = O((\eps \bar{n})^2)
= O((\eps n)^2)$.  As the variances of the previous rounds are
geometrically decreasing, the total variance from all the rounds is still
bounded by $O((\eps n)^2)$, as desired.

\begin{theorem}
  There is an algorithm for the rank-tracking problem that, at any time,
  estimate the rank of any element within error $\eps n$ with probability
  at least $0.9$.  It uses $O\left(\frac{1}{\eps\sqrt{k}}
    \log^{1.5}\frac{1}{\eps} \log^{0.5}\frac{1}{\eps \sqrt{k}}\right)$
  space at each site with communication cost
  $O\left(\frac{\sqrt{k}}{\eps}\log N \log^{1.5} \frac{1}{\epsilon
      \sqrt{k}}\right)$.
\end{theorem}

\providecommand{\noopsort}[1]{}

\appendix

\begin{figure*}
\begin{center}
\includegraphics{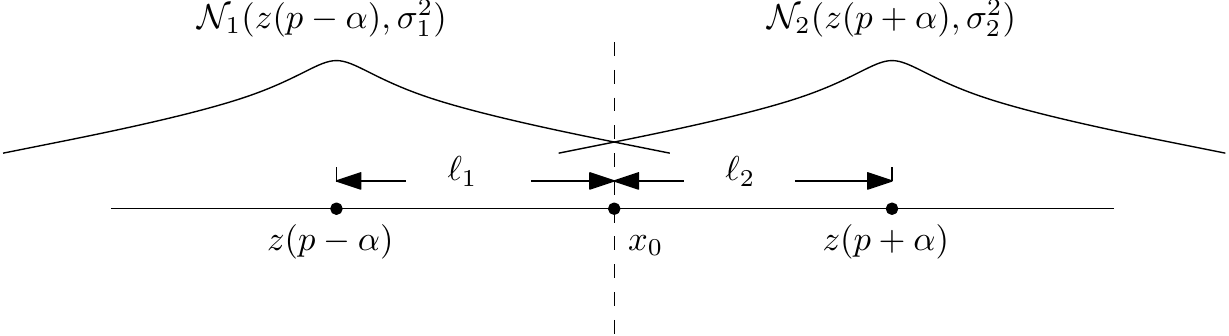}
\caption{Differentiating two distributions}
\label{fig:two-dist}
\label{fig:track-median}
\end{center}
\end{figure*}

\section{Lower bound for the sampling problem}
\begin{claim}
\label{cla:sampling}
To solve the sampling problem we need to probe at least $\Omega(k)$ sites.
\end{claim}

\begin{proof}
Suppose that the coordinator only samples $z = o(k)$ sites. Let $X$ be the number of sites that are sampled with bit $1$. Then $X$ is chosen from the hypergeometric distribution with probability density function (pdf) $\Pr[X = x] = {s' \choose x}{k' - s' \choose z - x}/{k' \choose z}$.
The expected value of $X$ is $\frac{z}{k'} \cdot s'$, which is
$\frac{z}{k'}\left(\frac{k}{2} - y + \sqrt{k}\right)$ or
$\frac{z}{k'}\left(\frac{k}{2} - y - \sqrt{k}\right)$, depending on the
value of $s'$. Let $p = \left(\frac{k}{2} - y\right)/k' = \frac{1}{2} \pm o(1)$ and $\alpha =
\sqrt{k}/k' = 1/\sqrt{k} \pm o(1/\sqrt{k})$. To avoid tedious calculation, we assume that $X$ is picked
randomly from one of the two normal distributions $\mathcal{N}_1(\mu_1,
\sigma_1^2)$ and $\mathcal{N}_2(\mu_2, \sigma_2^2)$ with equal probability,
where $\mu_1 = z(p-\alpha), \mu_2 = z(p+\alpha), \sigma_1, \sigma_2 =
\Theta(\sqrt{zp(1-p)}) = \Theta(\sqrt{z})$. In Feller~\cite{feller:68} it is shown that the normal distribution approximates the hypergeometric
distribution very well when $z$ is large and $p \pm \alpha$ are
constants in $(0,1)$~\footnote{In Feller's book~\cite{feller:68} the
  following is proved. Let $p \in (0,1)$ be some constant and $q =
  1-p$. The population size is $N$ and the sample size is $n$, so that $n <
  N$ and $Np, Nq$ are both integers. The hypergeometric distribution is
  $P(k; n, N) = {Np \choose k}{Nq \choose n-k}/{N \choose n}$ for $0 \le k
  \le n$.
\begin{theorem}\cite{feller:68}
If $N \to \infty, n \to \infty$ so that $n/N \to t\in(0,1)$ and $x_k := (k
- np)/\sqrt{npq} \to x$, then
$$p(k;n,N) \sim \frac{e^{-x^2/2(1-t)}}{\sqrt{2\pi npq(1-t)}}$$
\end{theorem}
}. Now our task is to decide from which of the two
distributions $X$ is drawn based on the value of $X$ with success
probability at least $0.7$.

Let $f_1(x; \mu_1, \sigma_1^2)$ and $f_2(x; \mu_2, \sigma_2^2)$ be the
pdf of the two normal distributions
$\mathcal{N}_1, \mathcal{N}_2$, respectively. It is easy to see that the
best deterministic algorithm of differentiating the two distributions based
on the value of a sample $X$ will do the following.
\begin{itemize}
\item If $X > x_0$, then $X$ is chosen from $\mathcal{N}_2$, otherwise $X$
  is chosen from $\mathcal{N}_1$, where $x_0$ is the value such that
  $f_1(x_0; \mu_1, \sigma_1^2) = f_2(x_0; \mu_2, \sigma_2^2)$ (thus $\mu_1
  < x_0 < \mu_2$).
\end{itemize}
Indeed, if $X > x_0$ and the the algorithm decides that ``$X$ is chosen from $\mathcal{N}_1$", we can always flip this decision and improve the success probability of the algorithm.

The error comes from two sources: (1) $X > x_0$ but $X$ is actually drawn
from $\mathcal{N}_2$; (2) $X \le x_0$ but $X$ is actually drawn from
$\mathcal{N}_1$. The total error is
$$1/2 \cdot (\Phi(-\ell_1/\sigma_1) + \Phi(-\ell_2/\sigma_2)),$$ where
$\ell_1 = x_0 - \mu_1$ and $\ell_2 = \mu_2 - x_0$. (Thus $\ell_1 + \ell_2 =
\mu_2 - \mu_1 = 2 \alpha z$). $\Phi(\cdot)$ is the cumulative distribution
function (cdf) of the normal distribution. See Figure~\ref{fig:two-dist}.

Finally note that $\ell_1/\sigma_1 = O(\alpha z / \sqrt{z}) = O(\sqrt{z/k})
= o(1)$ and $\ell_2/\sigma_2 = O(\alpha z / \sqrt{z}) = o(1)$, so
$\Phi(-\ell_1/\sigma_1) + \Phi(-\ell_2/\sigma_2) > 0.99$. Therefore, the
failure probability is at least $0.49$, contradicting our success probability guarantee.  Thus we must have $z = \Omega(k)$.
\end{proof}

\end{document}